\documentclass[11pt, a4paper]{article}
\usepackage[T2A]{fontenc}
\usepackage[cp1251]{inputenc}
\usepackage{amsmath,amssymb}
\usepackage{amsfonts}
\usepackage{graphicx,color}
\usepackage{amsthm}

\graphicspath{{./Figs/}}
\DeclareGraphicsExtensions{.pdf,.jpg,.png,.eps,.gif}
\date{\today}
\newcommand{\EQ}{\begin{equation}}
\newcommand{\EN}{\end{equation}}

\newtheorem{rem}{Remark}
\newtheorem{theo}{Theorem}
\newtheorem{deff}{Definition}
\newtheorem{coro}{Corollary}

\newtheorem{lem}{Lemma}
\newtheorem{conj}{Conjecture}

\newcommand{\pr}{\indent{\bf Proof. \ }}
\newcommand{\wt}{\mbox{\rm wt}}

\newcommand{\bx}{{\bf x}}

\newcommand{\by}{{\bf y}}
\newcommand{\bz}{{\bf z}}

\newcommand{\supp}{\mbox{\rm supp}}

\newcommand{\F}{{\mathbb F}}
\newcommand{\PG}{{\mathbb PG}}

\makeatletter % @ is now a letter
\def\bibleftdelim{}
\def\bibrightdelim{.}
\def\@biblabel#1{\bibleftdelim #1\bibrightdelim}
\makeatother % @ is a symbol

\title{ On two-weight codes}

\author{P. Boyvalenkov$^*$, K. Delchev\footnote{Institute of Mathematics and Informatics,
Bulgarian Academy of Sciences, 8 G. Bonchev Str., 1113  Sofia, Bulgaria
(e-mail: peter@math.bas.bg, math\_k\_delchev@yahoo.com).},
D. V. Zinoviev$^\dagger$, V. A. Zinoviev\footnote{A.A. Kharkevich
Institute for Problems of Information Transmission, Russian
Academy of Sciences, Bol'shoi Karetnyi per. 19, GSP-4, Moscow,
101447, Russia (e-mail: \{dzinov,vazinov\}@iitp.ru)}}

\date{}

\begin{document}
\maketitle

\begin{abstract}
We consider $q$-ary (linear and nonlinear) block codes with
exactly two distances: $d$ and $d+\delta$. We derive necessary
conditions for existence of such codes (similar to the known conditions in
the projective case). In the linear (but not necessary projective) case, we prove
that under certain conditions the existence of such linear
$2$-weight code with $\delta > 1$ implies the following equality
of greatest common divisors: $(d,q) = (\delta,q)$.  Upper bounds for
the maximum cardinality of such codes are derived by linear
programming and from few-distance spherical codes. Tables of lower
and upper bounds for small $q = 2,3,4$ and $q\,n < 50$ are
presented.
\end{abstract}


{\bf AMSsubject classification.} 94B65, 94B05

\section{Introduction}

Let $E_q = \{0,1, \ldots, q-1\}$, where $q \geq 2$ is a positive integer.
Any subset $C \subseteq E_q^n$ is called a code and
denoted by $(n,N,d)_q$; i.e., a code of length $n$, cardinality
$N = |C|$ and minimum (Hamming) distance $d$. For linear codes we use
notation $[n,k,d]_q$ (i.e., $N=q^k$). An $(n,N,d)_q$ code $C$ is
equidistant if for any two distinct codewords $x$ and $y$ we have
$d(x,y) = d$, where $d(x,y)$ is the (Hamming) distance between $x$
and $y$. A code $C$ is constant weight and denoted $(n,N,w,d)_q$
if every codeword is of weight $w$. For the binary case, i.e. when
$q=2$ we omit $q$ and use the notations $(n,N,d)$ and $[n,k,d]$.

We consider codes with only two distances $d$ and $d+\delta$. Such
codes have been a classical object in algebraic coding theory for
more than 50 years. A comprehensive survey of such
linear projective (i.e. when $n \leq (q^k-1)/(q-1)$) codes
can be found in the paper of Calderbank and Kantor \cite{CK86}.
Nevertheless in spite of many infinite classes of two-weight codes
the complete classification of such codes is far from being
completed. Our goal here is to understand the structure of
arbitrary (i.e. not only projective) two-weight codes and to
consider the general properties of all
such codes. We believe that for many possible values of $\delta$
such (linear and nonlinear) codes of dimensions larger than 2 do
not exist. In particular, we prove that if there exist a $q$-ary
linear code $C$ with two distances $d$ and $d+\delta$
where $\delta > 1$, then either $(q,d) = (q,\delta)$ or
$(q,d_c) = (q,\delta)$, where $d_c$ is the minimum distance
of complementary code $C_c$ with two distances $d_c$ and $d_c + \delta$,
which coexists with the code $C$. This generalizes previous results
of Delsarte for projective codes to arbitrary linear two-weight codes
\cite{Del72}. The case $\delta = 1$ was considered in our previous paper
\cite{BDZZ20}, where we classified all linear codes with distances $d$
and $d+1$ (see also \cite{LRS19}), derived upper bounds for the maximum possible cardinality in this case
and presented tables for the maximal possible cardinality for small alphabets and lengths.
Here we also give lower and upper bounds for maximum cardinality of codes
with two distances $d$ and $d+\delta$ and give tables of such linear and nonlinear codes.

Denote by $(n,N,\{d, d+\delta\})_q$ an $(n,N,d)_q$ code $C
\subset E_q^n$ with the property under investigation: for any two distinct
codewords $x$ and $y$ from $C$ we have $d(x,y) \in \{d,
d+\delta\}$.  We are interested in existence,  constructions,
and classification results and lower and upper bounds on the maximal
possible size of $(n,N,\{d, d+\delta\})_q$ codes. If $q$ is a prime
power, then $E_q$ will be the set of the elements of the Galois
field $\F_q$. In this case, if an $(n,N,d)_q$-code $C$ is a $k$-dimensional subspace
of the linear space $\F_q^n$, then we use for $C$ the standard notation
$[n,k,d]_q$, where $N = q^k$, and a two-weight $(n,N,\{d, d+\delta\})_q$-code
$C$ will be denoted as $[n,k,\{d, d+\delta\}]_q$.

\begin{deff}
A two-weight $(n,N,\{d, d+\delta\})_q$-code $C$ is called trivial,
if it satisfies one of the following properties:\\
(1) $C$ contains trivial positions, i.e. all its codewords
contain the same symbol in some position;\\
(2) $C$ can be presented as a concatenation of several two-weight
codes with the same parameters:
\[
C = (C_1\,|\,\cdots \,|\,C_s) = \{c^{(i)}=(c^{(i)}_1\,|\,\cdots \,|\,c^{(i)}_s):\,
c^{(i)}\in C,\;c^{(i)}_j \in C_j,\; j = 1, \ldots , s,\;i \in \{1,\ldots, N\}\},
\]
where every code $C_j$ is a two-weight code
with the same parameters, i.e. $C$ is an $(sn,N,\{sd, sd + s\delta\})_q$-code,
where $C_j$ is an $(n,N,\{d, d + \delta\})_q$-code for every $j \in \{1, \ldots, s\}$.
\end{deff}

Recall that all linear codes with two distances $d$ and $d+1$ were
described in our recent papers \cite{BDZZ18, BDZZ20} (see also
\cite{LRS19}). The above mentioned detailed survey of Calderbank
and Kantor \cite{CK86} gives the complete state (for that time) of
this subject mostly in terms of geometric concepts. Here we use
mostly combinatorial and coding-theoretical arguments.

\section{Preliminary results}

For two codes $A$ and $B = \{\by_j:\;j = 0,1,\ldots, N_b-1\}$ with
parameters $(n_a,N_a,d_a)_{q_a}$ and $(n_b,N_b,d_b)_{q_b}$, such
that $E_{q_b} \subseteq E_{q_a}$ and $N_b = q_a$, define a new
code $C$ over $E_{q_b}$ (which is called a {\em concatenated code},
or a {\em concatenation of $A$ and $B$}), such that
\[
C = \{ (\by_{x_1},\by_{x_2}, \cdots, \by_{x_{n_a}}):\;\bx =
(x_1,x_2, \cdots,x_{n_a}) \in A\},
\]
where every symbol $i \in E_{q_a}$ of codewords of $A$ we change by
codewords  $\by_i$ of $B$ (with index $i$).
The code $C$ has parameters $[n,N,d]_q$, where
\EQ\label{eq:2.0}
n~=~n_a\,n_b,\;\;d \geq d_a\,d_b,\;\;N~=~N_a,\;\;q=q_b~.
\EN

\begin{deff}\label{deff:dm}
Let $G$ be an abelian group of order $q$ written additively. A square
matrix $D$ of order $q\mu$ with elements from $G$ is called a difference matrix and denoted
$D(q, \mu)$, if the component-wise difference of any two different
rows of $D$ contains any element of $G$ exactly
$\mu$ times.
\end{deff}

See \cite{BJL86} for difference matrices. From \cite{SZZ69} we
have the following result.

\begin{lem}\label{lm:2.2}
For any prime number $p$ and any natural numbers $\ell$ and $h$ there
exists a difference matrix $D(p^\ell, p^h)$.
\end{lem}

Clearly the difference matrix $D(q, \mu)$ induces an equidistant $(q\mu-1,
q\mu, \mu(q-1))_q$ code  and also a nonlinear two-weight  code with the
following parameters \cite{SZZ69}:
\EQ\label{eq:2.1}
n = q\mu,\;\;N = q^2\mu,\;\; w_1 = (q - 1)\mu,\;\;w_2 = n.
\EN
These codes are optimal
according to the following $q$-ary Gray-Rankin bound
\cite{BDZH06,DHZ04}. Any $q$-ary $(n,N,d)_q$-code, whose
codewords can be partitioned into trivial subcodes $(n,q,n)_q$ (we
call such codes antipodal), has cardinality $N$ such that
\EQ\label{eq:2.3}
\frac{N}{q} \leq  \frac {q(qd - (q-2)n)(n-d)}{n-((q-1)n-qd)^{2}},
\EN
under condition that $n-((q-1)n-qd)^{2} > 0$. Note that this
bound is a $q$-ary analog
of the following classical Gray-Rankin bound for a binary
antipodal $(n,N,d)$-code $C$
\EQ\label{eq:2.30}
N \leq \frac{8d(n-d)}{n - (n-2d)^2}\,.
\EN

\section{Necessary conditions}

The natural question for existence of a $q$-ary  two-weight
$(n,N,\{d, d+\delta\})_q$-code is under which conditions such code
exist, if we fix, for example, a prime power $q$, the minimum
distance $d$ and cardinality $N$. The full answer for this question
is open.

Let $\PG(n,q)$ denote the $n$-dimensional projective space over the
field $\F_q$. A {\em $m$-arc} of points in $\PG(n,q)$,\, $m \geq n+1$ and
$n \geq 2$, is a set $M$ of $m$ points such that no $n+1$ points of
$M$ belong to a hyperplane of $\PG(n,q)$. The $(q+1)$-arc of $\PG(2,q)$
are called {\em ovals}, the $(q+2)$-arcs of $\PG(2,q)$,\,$q$ even, are called
{\em complete ovals}  or {\em hyperovals} (see, for example, \cite{Den69,Tha95}).

A linear code $C$ is called {\em projective} if its dual code
$C^\perp$ has minumum distance $d^\perp \geq 3$ (i.e., any
generator matrix of $C$ does not contain two columns that are
scalar multiples of each other). For projective  $[n,k,d]_q$-codes $C$
one can define the concept of {\em complementary code} (see, for
example, \cite{CK86}).

Let $[C]$ denote the matrix formed by all codewords of $C$ (i.e., the rows of $[C]$ are codewords of $C$).
The code $C_c$ is called a complementary of $C$, if the matrix
$[[C]\,|\,[C_c]]$ is a linear equidistant code and $C_c$ has the
minimal possible length which gives this property.

The extension of this well known concept to arbitrary linear two-weight
codes is formulated as the following evident lemma. Somehow, it is
connected with {\em anticodes} due to Farrel \cite{Far70} and {\em minihypers}
(see \cite{HH01}).

\begin{lem}\label{lm:4.1}
Let $C$ be a $q$-ary linear nontrivial two-weight $[n,k,\{d, d+\delta\}]_q$-code
and let $\mu_1$ and $\mu_2$ denote the number of codewords of weight $d$
and $d + \delta$, respectively. Then there exist the complementary
linear two-weight $[n_c,k,\{d_c, d_c+\delta\}]_q$-code $C_c$, where
\[
n + n_c = s\,\frac{q^k-1}{q-1},\;\;d + d_c + \delta = s q^{k-1},\;\;s=1,2, \ldots ,
\]
and where $C_c$ contains $\mu_1$ codewords of weight $d_c+\delta$ and
$\mu_2$ codewords of weight $d_c$ and where $C_c$ is of minimal possible
length, such that the matrix $[[C]\,|\,[C_c]]$ is an equidistant code.
\end{lem}

Note that the integer $s$ in the Lemma \ref{lm:4.1} is the maximal
size of the collection of columns in the generator matrix of $C$
which are scalar multiples of one column.
For projective two-weight codes (i.e. for the case $s=1$) the
following results are known.

\begin{lem}\label{lm:4.0} $\cite{Del72}$
Let $C$ be a two-weight projective $[n,k,\{w_1, w_2\}]_q$ code over
$\F_q$,\,$q = p^m$,\, $p$ is prime. Then there exist two integers
$u \geq 0$ and $h \geq 1$,
such that
\EQ\label{eq:3.00}
w_1 = h\,p^u,\;\;w_2 = (h+1)\,p^u.
\EN
\end{lem}

For the projective case, we recall the following result
(which directly follows from the MacWillams identities, taking
into account that the dual code $C^\perp$ has minimum
distance $d^\perp \geq 3$) (see \cite{Del72}).

\begin{lem}\label{lm:4.2}
Let $C$ be a $2$-weight projective $[n,k,\{w_1, w_2\}]_q$ code
$C$ over $\F_q$,\,$q = p^m$,\, $p$ is prime. Denote by $\mu_1$
the number of codewords of $C$ of weight $w_1$ and by $\mu_2$ the
number of codewords of weight $w_2$. Then
\EQ \label{lm4-eq1}
w_1\,\mu_1 + w_2\,\mu_2 = n(q-1)q^{k-1},
\end{equation}
\begin{equation} \label{lm4-eq2}
w_1^2\,\mu_1  + w_2^2\,\mu_2 = n(q-1)(n(q-1)+1)q^{k-2}.
\EN
\end{lem}

To formulate the next necessary condition of existence a projective two-weight
code, the so-called {\em integrality condition}, we need {\em strongly regular graphs}.

\begin{deff}\label{deff:srg}
A connected simple (i.e., undirected, without loops and multiple edges) graph
$\Gamma$ on $v$ vertices is called strongly regular with parameters $(N,K,\lambda,\mu)$
if it is regular with valency $K$ and if the number of vertices joined to two
given vertices is $\lambda$ or $\mu$ according as the two given vertices are
adjacent or non-adjacent.
\end{deff}

The adjacency matrix $A=[a_{ij}]$ of $\Gamma$ is a $(N \times N)$-matrix
with $a_{ij}=1$ if the vertices with indices $i$ and $j$ are adjacent and
$a_{ij}=0$, otherwise. The eigenvalues of $A$ are $K$, $\rho_1$ and $\rho_2$, where
\EQ\label{eq:3.10}
\rho_1, \rho_2 = \frac{1}{2}\left((\lambda - \mu) \pm \sqrt{\Delta}\right),\;\;
\Delta = (\lambda-\mu)^2 + 4(K-\mu).
\EN
The multiplicity of $K$, $\rho_1$ and $\rho_2$ are $1$,\,$e_1$, and $e_2$, respectively,
where
\EQ\label{eq:3.11}
e_1, e_2 = \frac{1}{2}\left(N-1 \pm \frac{(N-1)(\lambda-\mu) + 2K}{\sqrt{\Delta}}\right).
\EN
The equation (\ref{eq:3.10}) is called the {\em integrality} or {\em rationality condition}
because $e_1$ and $e_2$ must be integers. According to Delsarte \cite{Del72} a two-weight
projective $[n,k,\{d,d+\delta\}]_q$-code induces a strongly regular graph
with parameters $(N,K,\lambda,\mu)$. For two-weight codes the parameters $N, K, \lambda$
and $\mu$ are related to the code parameters $n, k, d = w_1$ and $w_2 = d+\delta$ as follows \cite{CK86}:
\EQ\label{eq:3.12}
\left.
\begin{array}{ccl}
N &=& q^k,\\
K &=& n(q-1),\\
\lambda &=& K(K + 3) - q(w_1+w_2)(K+1) + q^2w_1w_2,\\
\mu &=& K(K+1) - Kq(w_1+w_2) + q^2w_1w_2.
\end{array}
\right\}
\EN

Using (\ref{eq:3.12}) we obtain the following formulas for $\Delta$:
\EQ\label{eq:3.13}
\Delta = q^2(w_1-w_2)^2 = (q \delta)^2,
\EN
i.e., for the case of two-weight codes, the value $\sqrt{\Delta}$ is
always integer. The multiplicities $e_1, e_2$ in (\ref{eq:3.11})
of the eigenvalues, i.e. the integrality condition, take the following
form:
\EQ\label{eq:3.14}
e_1, e_2 = \frac{1}{2}\left(q^k - 1 \pm \frac{q^k(2n(q-1) - q(w_1+w_2))}{q(w_1-w_2)}\right)\,.
\EN

Here  we derive an integrality condition similar to
(\ref{eq:3.11}) (which differs from the conditions (\ref{eq:3.13})
and (\ref{eq:3.14})) for projective two-weight codes using simple
combinatorial arguments not connected with eigenvalues of strongly
regular graphs.

The next statement gives a slightly different quadratic equation, which might
be useful, since it is valid not only for linear projective codes. Our
derivation uses only simple combinatorial arguments.  The analogous statement
was proved and used in \cite{BDZZ20} for the special case $\delta = 1$.
Recall that a $q$-ary matrix $M$ of size $N \times n$ is called an
orthogonal array denoted $OA(N,n,q,t)$ of strength $t$, index
$\lambda = N/q^t$ and $n$ constraints, if any its submatrix of size
$N \times t$ contains as a row any $q$-ary vector of length $t$ exactly $\lambda$
times (see \cite{BJL86}).

\begin{theo}\label{th:4.1}
Let $E_q$ be any alphabet and let $C$ be a $q$-ary two-weight
$(n,N,\{w_1, w_2\})_q$-code which is also an orthogonal array of strength
$t \geq 2$, i.e. the value $N/q^2$ is an integer. Let $u_i = n - w_i$,\,
$i = 1,2$. Then
\begin{itemize}
\item The length $n$ of the code $C$ satisfies the following quadratic
equation:
\EQ\label{eq:3.22}
n^2 - n\left(Q_1(u_1+u_2-1)+1\right) + Q_2 u_1u_2 = 0,\;\;
Q_1=q\,\frac{N-q}{N-q^2},\;\;Q_2=q^2\,\frac{N-1}{N-q^2}.
\EN
\item If $w_2 = n$ (i.e. if $u_2 = 0$), then the length $n$ of the code $C$ satisfies
\EQ\label{eq:3.22a}
n = \frac{Q_1(d+1) - 1}{Q_1-1},
\EN
where $d = w_1$ is the minimum (Hamming) distance of $C$.
\end{itemize}
\end{theo}

\pr
Note that if $C$ is linear, then its length $n$ should be less than
$(q^k-1)/(q-1)$, i.e. the code $C$ should be projective. Denote by $[C^*]$
the $(N-1 \times n)$-matrix formed by the all $N-1$ nonzero codewords of
$(n, N, \{w_1,w_2\})_q$-code $C$ (remark that $C$ is not necessary linear).
Let $\mu_i$ denote the number of codewords of weight $w_i$ in $C$ where $i=1,2$
and $w_2 > w_1$. It is convenient to work with number of zeroes in codewords
instead of the weight. So, denote by $u_i = n - w_i$ the number of zeroes in
a codeword of weight $w_i$.  Denote by $\Sigma_{(0)}$ the overall amount of zeroes
in the matrix $[C^*]$. For the overall number $\Sigma_{(0)}$ of
zero positions in $[C^*]$ we evidently have
\[
\Sigma_{(0)} = n (\frac{N}{q} - 1).
\]
On the other side this number  equals
\[
\Sigma_{(0)} = \mu_1 u_1 + \mu_2 u_2.
\]
We deduce, taking into account the evident equality $\mu_1 + \mu_2 = N - 1$, that
\EQ\label{eq:3.23}
 \mu_1 (u_1 - u_2) = n \frac{N-q}{q} - u_2 (N - 1).
\EN

Since all codewords of $C$ form  an orthogonal array of strength $t \geq 2$,
we can compute the number of all zeroes pairs $(0,0)$ (denoted by $\Sigma_{(0,0)}$)
in the matrix $[C^*]$ in two different ways. First, we have
\[
\Sigma_{(0,0)} = \binom{n}{2} (\frac{N}{q^2}-1)
\]
(recall that $[C^*]$ does not contain the zero codeword).
On the other side, since there are $\mu_1$ codewords with $u_1$ zero positions
and $\mu_2$ codewords with $u_2$ such positions we obtain
\[
\Sigma_{(0,0)} = \binom{u_1}{2} \mu_1 + \binom{u_2}{2} \mu_2 =
\left(\binom{u_1}{2}-\binom{u_2}{2}\right) \mu_1 + \binom{u_2}{2} (N-1).
\]
Taking into account (\ref{eq:3.23}) and using the equality
\[
\binom{u_1}{2}-\binom{u_2}{2} = \frac{1}{2}\,(u_1-u_2)(u_1+u_2-1),
\]
we arrive to the equality
\[
\binom{n}{2}\,\frac{N-q^2}{q^2} = \left(n\,\frac{N-q}{q} - u_2(N-1)\right) \frac{1}{2}\,(u_1+u_2-1)
+\binom{u_2}{2}(N-1),
\]
which reduces to the quadratic equation \eqref{eq:3.22}.

Hence for the existence of the two-weight $(n,N,\{w_1,w_2\})_q$-code
$C$ (which is also an orthogonal array of strength $t \geq 2$) the
equation (\ref{eq:3.22}) should have a solution in positive integer $n$ (where $u_i = n - w_i$,\,$i=1,2$).
Therefore
%\EQ\label{eq:3.25}
\[ n_1, n_2 = \frac{1}{2}\left((Q_1(u_1+u_2-1)+1) \pm \sqrt{(Q_1(u_1+u_2-1)+1)^2 - 4 Q_2 u_1u_2}\right) \]
%\EN
should be positive integers (here $Q_1$ and $Q_2$ are defined in (\ref{eq:3.22})).
To have a solution, first, the following
inequality should be satisfied:
%\EQ\label{eq:3.26}
\[ 4 Q_2 u_1u_2 \leq (Q_1(u_1+u_2-1)+1)^2. \]
%\EN
Second, the number
%\EQ\label{eq:3.27}
\[ (Q_1(u_1+u_2-1)+1)^2 - 4 Q_2 u_1u_2 \]
%\EN
should be a perfect square. This finishes the proof.

In the case $w_2 = n$, i.e. $u_2 = 0$, we can put the value $u_2 = 0$ into (\ref{eq:3.22}). This gives the expression \eqref{eq:3.22a} for $n$.
%\EQ\label{eq:3.27}
%n = q(u_1 - 1)\,\frac{N-q}{N-q^2} + 1.
%\EN
\qed

\begin{rem}\label{rem:2.1}
Note that counting of distinct
pairs of nonzero positions (i.e. pairs $(a,b)$,\, $a,b \neq 0$) in all codewords
of the code $C$ gives exactly the expression for $n$, which can be deduced from
(\ref{eq:3.12}).
\end{rem}

The codes with parameters (\ref{eq:2.1}) induced by difference matrices, have $w_2 = n$
and, hence, their parameters satisfy the formula (\ref{eq:3.22a}). Another large class
of two-weight codes with $w_2=n$ are the codes, constructed by Bush \cite{B52}. 
These codes are linear and have parameters: $n = q + 2$,\;
$k = 3$,\;$w_1 = q$,\;$w_2 = n$, where $q = 2^s$.
Codes with $w_2 = n$ are associated with the hyperovals. In fact, all maximal arcs give rise to linear two-weight codes
with parameters:
\[
n = t(q + 1) - q,\;\; k = 3,\;\;w_1 = n - t,\;\;w_2 = n,\;\; q = 2^s.
\]
The examples which we give above correspond to the case $t = 2$
(see \cite{Den69}).

In the next statement we generalize Lemma \ref{lm:4.0} to the case of arbitrary two-weight
$[n,k,\{d,d+\delta\}]_q$-codes. Besides, we obtain slightly stronger result for
projective codes. Here we assume that $q = p^m$ where $m \geq 1$ and
$p$ is prime. Given $q=p^m$ and an arbitrary positive integer $a$ denote
by $\gamma_a \geq 0$ the maximal integer, such that $p^{\gamma_a}$ divides $a$,
i.e. $a = p^{\gamma_a}\,h$, where $h$ and $p$ are co-prime.
Let $\gamma_d$, $\gamma_\delta$ and $\gamma_c$ be defined similarly for $d$,
$\delta$ and $d_c$, respectively. Recall that $(a,b)$ denote the greatest
common divisor (gcd) of two integers $a$ and $b$.

\begin{theo}\label{th:4.2}
Let $q = p^m$, where $m \geq 1$ and $p$ prime. Let $C$ be a $q$-ary linear
nontrivial two-weight
$[n,k,\{d, d+\delta\}]_q$-code of dimension $k \geq 2$ and let $C_c$ be its
complementary two-weight $[n_c,k,\{d_c, d_c+\delta\}]_q$-code $C_c$, where
\EQ\label{eq:4.00}
d + d_c + \delta = s\,q^{k-1},\;\;s \geq 1\,.
\EN
\begin{itemize}
\item[(i)] If $s = 1$ and $k \geq 4$, i.e. $C$ and hence $C_c$ are projective
codes, then the following two equalities  are satisfied:
\EQ\label{eq:4.0}
(q,d) \,=\,(q,\delta)\;\;\mbox{and}\;\;(q,d_c) \,=\,(q,\delta)\,.
\EN

\item[(ii)] If $s = 1$ and $k = 3$, then both equalities in (\ref{eq:4.0}) are
satisfied, if at least one  of the following two conditions takes place:
\[
(d,q)^2 \leq q (n(n-1),q)\;\;\mbox{or}\;\;(d+\delta,q)^2 > q (n_c(n_c-1),q)\,.
\]
\item[(iii)] If $s = 1$ and $k \geq 2$, then at least one of the following two
equalities is satisfied:
\EQ\label{eq:4.00}
\gamma_d = \gamma_\delta,\;\;\mbox{or}\;\;\gamma_c = \gamma_\delta\,.
\EN
\item[(iv)] If $s \geq 1$ and $k \geq 3$, then at least one of the two equalities
in (\ref{eq:4.00}) (respectively, in (\ref{eq:4.0})) is valid.

\end{itemize}
\end{theo}

\pr
 We start from the statement (iv).
Let $C$ be a $q$-ary linear two-weight $[n,k,\{d, d+\delta\}]_q$-code.
Recall that $\mu_1$ is the number of codewords of $C$ of weight $d$ and
$\mu_2$ is the number of codewords of weight
$d+\delta$. Then \eqref{lm4-eq1} and the evident equality
\EQ\label{eq:4.2}
\mu_1 + \mu_2 = q^k - 1,
\EN
imply that
\EQ\label{eq:4.3}
(q^k-1) d + \mu_2\delta = n(q-1)q^{k-1}.
\EN
We deduce (recall that $k \geq 3$) from the equality (\ref{eq:4.3}) that
$\gamma_\delta \leq \gamma_d$ and $\gamma_d \leq \gamma_\mu + \gamma_\delta$
(where $\mu_2 = h q^{\gamma_\mu}$ and $q$ and $h$ are co-prime).
If $\gamma_\mu = 0$, then we obtain $\gamma_d = \gamma_\delta$, otherwise we
have $\gamma_d > \gamma_\delta$.

Consider the case $\gamma_\mu \geq 1$ (or equivalently, $(\mu_2,q) > 1$).
By Lemma \ref{lm:4.1},
the existence of $C$ implies the existence of the complementary
two-weight $[n_c,k,\{d_c, d_c+\delta\}]_q$-code $C_c$, containing $\mu_1$
codewords of weight $d_c +\delta$ and $\mu_2$ codewords of weight $d_c$.
The equation (\ref{eq:4.3}) for the code $C_c$ looks as
\EQ\label{eq:4.5}
(q^k-1) d_c + \mu_1\delta = n_c(q-1)q^{k-1}
\EN
(indeed, in $C_c$ the codewords of weight $d_c+\delta$ occur $\mu_1$ times).
Taking into account that $(\mu_2,q) > 1$ and $\mu_1 + \mu_2 \equiv -\,1 \pmod{q}$
from \eqref{eq:4.2},
we deduce that $(\mu_1,q) = 1$. Hence from (\ref{eq:4.5}) we obtain that
\EQ\label{eq:4.02}
\gamma_c = \gamma_\delta.
\EN
From the equalities in (\ref{eq:4.00}) between $\gamma$'s of parameters $d$, $d_c$
and $\delta$, we obtain the corresponding equalities in (\ref{eq:4.0})
between the corresponding greatest common divisors. This completes the proof of (iv).

For (i), i.e. for the case $s=1$ and $k \geq 4$, we already know
that one of the equalities in (\ref{eq:4.0}) is valid. Let us assume for a
contradiction that $(q,d_c)=(q,\delta)$ but $(q,d)>(q,\delta)$.
The identity \eqref{lm4-eq2} can be written as follows:
\EQ\label{eq:4.04}
(\mu_1 + \mu_2) d^2 + 2\mu_2 d\delta + \mu_2 \delta^2 = n(q-1)(n(q-1)+1)q^{k-2}\,.
\EN

Taking into account the equalities (\ref{eq:4.00}), (\ref{eq:4.3}),
and (\ref{eq:4.02}), set
\EQ\label{eq:4.06}
(d,q) = p^t,\;\;(d_c,q) = (\delta,q) = p^r,\;\;(\mu_2,q) = p^u,
\EN
where $u, r \geq 1$ and $t = r + u$. Then (\ref{eq:4.04}) gives a contradiction
modulo $p^{t+r+1}$ for any $k \geq 4$ and any $t \leq m$ (recall that $q = p^m$
and $p$ is prime). Indeed, only the third (from the left) monom in (\ref{eq:4.04})
is nonzero  modulo $p^{t+r+1}$.

Consider the case $k = 3$, i.e. the statement (ii).
Clearly we obtain the same contradiction in the equality (\ref{eq:4.04})
(i.e., the same third monom at the left would be only nonzero modulo $p^{t+r+1}$)
for any $t \leq m/2$. Furthermore, because of the following equality for greatest common divisors,
\[
(n(q-1)(n(q-1)+1),q) = (n(n-1),q),
\]
we arrive to the same contradiction in (\ref{eq:4.04}) for the case when
\[
(d,q)^2 \leq q(n(n-1),q).
\]

Clearly the same idea can be used for the complementary (projective) code $C_c$.
The analog of \eqref{eq:4.04} for $C_c$ is
\EQ\label{eq:4.05}
(\mu_1 + \mu_2) d_c^2 + 2\mu_1 d_c\delta + \mu_1 \delta^2 = n_c(q-1)(n_c(q-1)+1)q^{k-2}\,.
\EN
We obtain a contradiction, if the following inequality would be valid:
\[
(d_c,q)^2 > q (n_c(n_c-1),q).
\]
Then the left hand side is divisible by $(d_c,q)^2$ and the right hand side is not (note that $k=3)$. This implies (ii).

For the case (iii), according to Lemma \ref{lm:4.0} there exist
nonnegative integers $g$ and $\gamma$ such that

\EQ\label{eq:4.08}
d = g\,p^{\gamma},\;\;\mbox{and}\;\;d+\delta = (g+1)\,p^{\gamma}.
\EN
First, assume that $g = \ell p^\alpha$, where $\ell$ and $p$ are
mutually prime. Using (\ref{eq:4.00}), we obtain
\[
d_c = sq^{k-1} - \ell p^\alpha p^\gamma - p^\gamma =
p^\gamma (s p^{m(k-1)-\ell \gamma} - p^\alpha - 1),
\]
implying for this case that
\[
\gamma_d \neq \gamma_\delta,\;\;\mbox{ but}\;\; \gamma_c = \gamma_\delta.
\]

For the case $g+1 = \ell p^\alpha$ similar arguments  imply that
\[
\gamma_d = \gamma_\delta,\;\;\mbox{but}\;\; \gamma_c \neq \gamma_\delta.
\]

Now assume that $(g,p) = 1$ and $(g+1,p) = 1$. We obtain for this case
\[
d_c = p^\gamma \left(sp^{m(k-1)-\gamma} - (g + 1)\right)
\]
implying that $\gamma_c = \gamma$. Since $\gamma_d = \gamma_\delta = \gamma$,
we obtain both equalities:
\[
\gamma_d = \gamma_\delta = \gamma_c.
\]
\qed

Note that the condition $k \geq 3$ in the cases (i) - (iii) of
Theorem \ref{th:4.2} can not be removed. It is easy to construct
two-weight $[n,2,\{d, d+\delta\}]_q$-code, where $\delta$ is an arbitrary
positive integer. Indeed, extend the well known equidistant
$[q+1,2,q]_q$-code $A$ with generating vectors $\bx_1$ and
$\bx_2$ as follows: add the zero vector $\bf{0}$ of length $\delta$
to $\bx_1$ and any vector $\bz$ of weight $\delta$ and length $\delta$
to $\bx_2$. The resulting two new codewords $\by_1 = (\bx_1\,|\,\bf{0})$
and $\by_2 = (\bx_2\,|\,\bz)$ generate a two-weight
$[q+1,2,\{q, q + \delta\}]_q$-code $C$, where  $\delta$ is an
arbitrary natural number (implying, in particularly,
that the equality $(d,q) = (\delta,q)$ is almost never valid). In this case check the
distance $d_c$ of the complementary code. The code $C$ has $s = \delta+1$.
Indeed, we add $\delta$ linearly dependent over $\F_q$ columns to the
matrix, which has already one such kind of column (up to multiplying by a
scalar). Hence, we have for
the distance $d_c$ of the complementary code $C_c$:
\[
d_c = (\delta+1)q - q - \delta = \delta (q-1).
\]
So we obtain for such codes that $(d,q) = q$ and $(\delta,q) = (d_c,q)$
is any natural number. Hence, the first equality in (\ref{eq:4.0})
is valid only for $\delta$ multiple to $q$, and the second equality in
(\ref{eq:4.0}) is valid always.

Let $q=p^m\geq 4$ be a prime power and $2 \leq r \leq q+1$.
From the MDS $[r,2,\{r-1,r\}]_q$-code $A$ and the
equidistant (simplex) code $B$ with parameters
\EQ\label{eq:4.10}
n_b = \frac{p^m-1}{p-1},\;\;k_b = m,\;\;d_b = p^{m-1},\;\; q_b = p,
\EN
we obtain by concatenation the family of
two-weight linear $p$-ary codes with the following parameters
(the family $SU2$ in \cite{CK86}):
\EQ\label{eq:4.12}
n = r\,\frac{p^m-1}{p-1},\;\; k = 2m,\;\;d =
(r-1)p^{m-1},\;\; \delta = p^{m-1},\;\;r=2,\ldots, q+1.
\EN
These $[n,2m, \{d, d+\delta\}]_p$-codes with
parameters (\ref{eq:4.12}) have $d = (r-1)p^{m-1}$ and $\delta =
p^{m-1}$ where $r \leq p^m+1$. Hence for $r = p^\ell+1
\leq p^m+1$ we obtain $\gamma_d = m+\ell-1$ and $\gamma_\delta =
m-1$ and the first equality in (\ref{eq:4.00}) is not valid.
Let us find the parameters of complementary code $C_c$. Clearly
\[
n_c = (p^m + 1 - r)\,\frac{p^m-1}{p-1},\;\;d_c = (p^m - r)\,p^{m-1},
\]
which implies $d_c = (p^m-p^\ell-1)\,p^{m-1}$ and hence $\gamma_c = m-1$.
Thus, $\gamma_c = \gamma_\delta$ and
the second equality in (\ref{eq:4.00}) is valid.

In some cases the conditions (\ref{eq:4.0}) and (\ref{eq:4.00}) are also
sufficient.

\begin{theo}\label{th:4.3}
Let $q = p^u$ be a prime power, $\delta = (q,\delta) h$ where $q$ and
$h$ be mutually prime, and let $s \geq 1$ be a natural number.
\begin{itemize}
\item[(i)] If $d + \delta = s\,q^r$ then for any $\delta = (q,\delta) h$, such that
$(q,d) = (q,\delta)$ and $h \leq s$, there exist a $q$-ary
linear two-weight $[n, r+1, \{d, d+\delta\}]_q$ code $C$ of length
\EQ\label{eq:4.20}
n = s\,\frac{q^{r+1}-1}{q-1} - h\,\frac{q^{\ell+1}-1}{q-1}\,.
\EN
If $h \leq q-1$, then the code $C$ is optimal as a $[n,k,d]_q$-code.
\item[(ii)] If $d = s\,q^r$ then for any $\delta = (q,\delta) h$, such that
$(q,d) = (q,\delta)$, there exist a $q$-ary linear
two-weight $[n, r+1, \{d, d+\delta\}]_q$ code $C$ of length
\EQ\label{eq:4.21}
n = s\,\frac{q^{r+1}-1}{q-1} + h\,\frac{q^{\ell+1}-1}{q-1}\,.
\EN
If $h = s$ then there exist an optimal
two-weight $(n, N, \{d, d+\delta\})_q$ code $C$ of length
$n = d +\delta$ and cardinality $N = q\,n$.
\item[(iii)] Let $p$ be any prime and $t$ be any natural number.
If $\gamma_d = t$ and $\delta = p^t$, i.e. $\gamma_d = \gamma_\delta$,
then for any $d = h\,p^{\gamma_d}$ where $h$ is any natural number
mutually prime to $p$, such that $h \leq p^{t+1} + 1$,
there exists a $p$-ary optimal two-weight $[n,k,\{d,d+\delta\}]_p$-code.
\end{itemize}
\end{theo}

\pr (i) Let $d + \delta = s\,q^r$, where $q = p^u$ and $u \geq 1$,
i.e. $q \geq p$. Under the conditions of the theorem, we can set
$\delta = h \,q^\ell$, where $1 \leq \ell \leq r-1$, and where
$(h,q) = 1$ and $h \leq s$.

Consider the linear equidistant $[n,k,d]_q$-code $A$
(dual to the $q$-ary Hamming code of length $n$) with
parameters:
\[
n = \frac{q^m-1}{q-1},\;\;k = m,\;\;d = q^{m-1}.
\]
It is well known that $A$ contains as a subcode a
linear $q$-ary equidistant code $B$ with parameters
\[
n_b = \frac{q^r-1}{q-1},\;\;k_b = r,\;\;d_b = q^{r-1},\;\;r=2,3, \ldots, m-1.
\]
Taking $s$ copies of $A$ and $h$ copies of $B$, we obtain
(by deleting $B$ from $A$) the family of linear $q$-ary
two-weight codes (the family $SU1$ in \cite{CK86}) with the following
parameters:
\EQ\label{eq:4.30}
n = \frac{s(q^m-1) -
h(q^r-1)}{q-1},\;\;k = m,\;\; d = s\,q^{m-1}-h\,q^{r-1},\;\;
\delta = h\,q^{r-1}\,,
\EN
where $r=2,\ldots, m-1$ and  $1 \leq h \leq s$. So, we obtain a
linear two-weight code $C$ of length (\ref{eq:4.20}),
which satisfies the condition of the theorem. If $h \leq q-1$,
then $C$ (as a $[n,k,d]_q$-code) is optimal according to the
Griesmer bound. This gives the first statement.

(ii) Consider the case $d = s\,q^r$. Assume that $\delta = h \,q^\ell$,
where $1 \leq \ell \leq r-1$, and where $h$ is any positive integer coprime to $q$. 
Consider the codes $A$ and $B$ from 
the case (i) above. Taking the union of $s$ copies of $A$ and
$h$ copies of $B$, we obtain a linear two-weight code $C$ of length
(\ref{eq:4.21})
which satisfies the condition of the theorem.

For the case when $d = s\,q^r$ and $\delta = s\,q^\ell$, where
$1 \leq \ell \leq r-1$, one can choose the codes with parameters (\ref{eq:2.1})
from difference matrices, which have the minimum possible length $n = d+\delta$
and cardinality $N = q\,n$. In this case the resulting code is nonlinear
until $n = q^u$. These codes are optimal since they attain the $q$-ary
Gray-Rankin bound (\ref{eq:2.3}).

(iii) Let $q=p^m\geq 4$ be a
prime power and $2 \leq r \leq q+1$.
From the outer MDS $[r,2,\{r-1,r\}]_q$-code $A$ and the inner
equidistant (simplex) code $B$ with parameters
\EQ\label{eq:4.31}
n_b = \frac{p^m-1}{p-1},\;\;k_b = m,\;\;d_b = p^{m-1},\;\; q_b = p,
\EN
we obtain  the following family of concatenated
two-weight linear $p$-ary codes with the following parameters
(the family $SU2$ in \cite{CK86}):
\EQ\label{eq:4.32}
n = r\,\frac{p^m-1}{p-1},\;\; k = 2m,\;\;d =
(r-1)p^{m-1},\;\; \delta = p^{m-1},\;\;r=2,\ldots, q+1.
\EN
Set $m = t+1$ and choose any $h = r \leq p^m + 1$,
which is mutually prime to $p$. So, for any such $h$ these codes have
$d = h\,p^{m-1}$ and $\delta = p^{m-1}$, such that $\gamma_d = \gamma_\delta$.
This gives (iii).
\qed

Theorem \ref{th:4.3} has a nice geometric explanation: it consists in removing $h$
copies of the same subspace $\PG(\ell, q)$ from $s$ copies of
projective space $\PG(r, q)$ (respectively,
adding $h$ copies of the same subspace $\PG(\ell, q)$ to $s$ copies of $\PG(r, q)$).
In the case when $r + 1 = 2(\ell + 1)$ we can get two-weight linear codes
by removing a partial spread consisting of $h$ subspaces of dimension $\ell$
from $\PG(r, q)$. Since in this case $\PG(r, q)$ can be partitioned into
$\ell$-subspaces, $h$ is unrestricted (see \cite{Tha95}).

\section{Upper bounds}

We are interested in upper bounds for the quantity
\[ A_q(n;\{d,d+\delta\}) = \max \{ |C| : C \mbox{ is an $(n,|C|, \{d,d+\delta\})$ code}\}, \]
the maximal possible cardinality of a code in $E_q^n$ with two distances $d$ and $d+\delta$.

\subsection{General linear programming bound}

We adapt the Delsarte linear programming bound for $A_q(n;\{d,d+\delta\})$.

For fixed $n$ and $q$, the (normalized) Krawtchouk polynomials are defined by
%%\begin{equation}\label{Kraw}
\[
Q_i^{(n,q)}(t) =\frac{1}{r_i} K_i^{(n,q)}(z), \ z=\frac{n(1-t)}{2},  \ r_i=(q-1)^i {n \choose i},
\]
%%\end{equation}
where
\[ K_i^{(n,q)}(z)=\sum_{j=0}^{i}   (-1)^j(q-1)^{i-j} {z \choose j} {n-z \choose i-j} \]
are the (usual) Krawtchouk polynomials.

If $f(t) \in \mathbb{R}[t]$ is of degree $m \geq 0$, then it can be uniquely expanded as
\begin{equation} \label{kraw-exp}
f(t) =    \sum_{i=0}^n f_i Q_i^{(n,q)}(t),
\end{equation}
where, if $\deg(f) \geq n+1$, the polynomial $f(t)$ is considered modulo $\prod_{i=0}^n (t-1+2i/n)$.

\begin{theo}
\label{thm lp}
Let $n \geq q \geq 2$ and $f(t)$ be a real polynomial such that:

{\rm (A1)} $f(t) \leq 0$ for $t \in  \{1-2d/n,1-2(d+\delta)/n\}$;

{\rm (A2)} the coefficients in the Krawtchouk expansion \eqref{kraw-exp} satisfy $f_0>0$ and $f_i \geq 0$ for every $i \geq 1$.

Then
\begin{equation} \label{lp-gen} A_q(n;\{d,d+\delta\}) \leq \frac{f(1)}{f_0}. \end{equation}
If an $(n,N,\{d,d+\delta\})_q$ code $C$ attains \eqref{lp-gen} for some polynomial $f(t)$, then
$f(1-2(d+i)/n)=0$, $i=0,\delta$, whenever there are points of $C$ at distance $d+i$, $i=0,\delta$, and
$f_iM_i(C)=0$, where
\begin{equation} \label{moments} M_i(C)= \sum_{x,y \in C} Q_i^{(n,q)}(1-2d(x,y)/n)=0
\end{equation}
is the $i$-th moment of $C$.
 \end{theo}

Proofs of such bounds are usually considered as folklore (see, for example, \cite{Del73,Lev95}). 
Applying (A1) and (A2) in the identity (see Equation (26) in \cite{Lev95})
\begin{equation} 
\label{main-id}
f(1)|C|+\sum_{x,y \in C, x \neq y} f(1-\frac{2d(x,y)}{n})=f_0|C|^2+\sum_{i=1}^n f_i M_i(C), 
\end{equation}
immediately implies Theorem \ref{thm lp}. We illustrate the use of \eqref{main-id} by another proof of Theorem  \ref{th:4.1}.

{\it Second proof of Theorem  \ref{th:4.1}}. We apply \eqref{main-id} with the polynomial
\[ f(t)=\left(t-1+\frac{2w_1}{n}\right)\left(t-1+\frac{2w_2}{n}\right) \]
and a code $C$ in the context of Theorem \ref{th:4.1}. Then $M_1(C)=M_2(C)=0$ since $C$ has strength 2, $f_i=0$ for $i \geq 3$,
and 
\[ \sum_{x,y \in C, x \neq y} f(1-2d(x,y)/n)=0 \] since $f$ vanishes at the "inner products"
$1-2d(x,y)/n$ of $C$. Therefore $f(1)=f_0N$. 

Since $f(1)=4w_1w_2/n^2$ and
\[ f_0={\frac {4(n^{2}(q-1)^2-n(q-1)(w_1+w_2q-1)+q^2w_1w_2)}{{n}^{2}{q}^{2}}},\]
after simplifications we obtain \eqref{eq:3.22}. 
\qed

\subsection{Specified linear programming bounds}

The degree one polynomial $f(t)=t-1+2d/n$ gives the Plotkin bound which is attained for many large $d$.
We proceed with degree two polynomials, where the bound produced coincides with the bound by Helleseth-Kl\o{}ve-Levenshtein \cite{HKL06}
for the maximal cardinality $|C|$ of a code $C$ with given minimum and maximum distances; this is also the bound for $k=1$
of Theorem 5.2 in \cite{BDHSS-ieee}. Here we give a proof which is direct from Theorem \ref{thm lp}.

\begin{theo}\label{thm lp1}
If
\begin{equation} \label{f1>0}
q(2d+\delta) \geq 2nq+2-2n-q,
\end{equation}
\begin{equation} \label{f0>0}
n(q-1)(nq-n+1)+nq(2d+\delta)>q^2(2nd+n\delta-d^2-d\delta),
\end{equation}
then
\begin{eqnarray}
\label{d2-bound}
 A_q(n,\{d,d+\delta\}) \leq \frac{d(d+\delta)q^2}{n(q-1)(nq-n+1)-q^2(2nd+n\delta-d^2-d\delta)+nq(2d+\delta)}.
\end{eqnarray}
If this bound is attained by an $(n,N,\{d,d+\delta\})_q$ code $C$, then $M_2(C)=0$ and,
moreover, $M_1(C)=0$ whenever \eqref{f1>0} is strict. In the latter case $C$ is an orthogonal array of strength 2.
\end{theo}

\begin{proof}
Consider the second degree polynomial
\[ f(t)=\left(t-1+\frac{2d}{n}\right)\left(t-1+\frac{2(d+\delta)}{n}\right). \]
The condition (A1) is obviously satisfied. For (A2), we find the Krawtchouk coefficients of $f(t)$ as follows
\begin{eqnarray*}
f_0 &=& \frac{4(n(q-1)(nq-n+1)-q^2(2nd+n\delta-d^2-d\delta)+nq(2d+\delta)}{n^2q^2}, \\
f_1 &=& \frac{4(q-1)(2dq+\delta q+2n+q-2nq-2)}{nq^2}, \\
f_2 &=& \frac{4(q-1)^2(n-1)}{nq^2}.
\end{eqnarray*}
It is obvious that $f_2>0$. Furthermore, $f_1 \geq 0$ and $f_0>0$ are equivalent to
\eqref{f1>0} and \eqref{f0>0}, respectively.
Therefore, provided \eqref{f1>0} and \eqref{f0>0}, we have
\[ A_q(n,\{d,d+\delta\}) \leq \frac{f(1)}{f_0}, \]
which gives the desired bound.
\end{proof}

If the right hand side of \eqref{d2-bound} is integer, we are able to find the distance distribution of $C$ by solving the system of equations
coming from  $A_d+A_{d+\delta}=|C|-1$ and $M_i(C)=0$, $i=1,2$. In the range of the tables this gives three nonexistence result, proving that
$A_2(12,\{6,10\}) \leq 19$ instead of 20, $A_2(20,\{10,14\}) \leq 27$ instead of 28, and $A_2(16,\{8,14\}) \leq 27$ instead of 28 from \eqref{d2-bound}.

\subsection{Upper bounds via spherical codes}

There is a natural relation between codes from $E_q^n$ and few-distance spherical codes.
First, the alphabet symbols $0,1,\ldots,q-1$ are mapped bijectively onto the vertices of the
regular simplex in $\mathbb{R}^{q-1}$. Then the codewords of any code $C \subset E_q^n$ can be send
(coordinate-wise) to $\mathbb{R}^{(q-1)n}$. It is not difficult to see that all obtained vectors have the
same length and after a normalization a spherical code $W \subset \mathbb{S}^{(q-1)n-1}$ is formed.

The code $W$ has the same cardinality as $C$, i.e., $|W|=|C|$, and its maximal inner product
is equal to $1-2dq/(q-1)n$, i.e., its squared minimum distance is $2dq/(q-1)n$. In our considerations,
the $q$-ary codes with distances $d$ and $d+\delta$ are mapped to spherical 2-distance codes with
squared distances $2dq/(q-1)n$ and $2(d+\delta)q/(q-1)n$. This implies a upper bound for
$A_q(n,\{d,d+\delta\})$ as follows.

\begin{theo}\label{thm sc}
Let $\frac{d}{d+\delta}=\frac{r}{s}$ in lowest terms. If $s-r \geq 2$ (in particular, if $(d,d+\delta)=1$)
or $s=r+1$ and $r>(\sqrt{2(q-1)n}-1)/2$, then
\[
A_q(n,\{d,d+\delta\}) \leq 2(q-1)n+1.
\]
\end{theo}

\pr A classical results by Larman, Rogers, and Seidel \cite{LRS77} states that if the cardinality of a 2-distance set $W \subset \mathbb{R}^m$
with distances $a$ and $b$, $a <b$, is greater than $2m + 3$, then the ratio $a^2/b^2$ is equal to $(k-1)/k$,
where $k \in [2,(\sqrt{2m}+1)/2]$ is a positive integer. The restriction $2m+3$ was moved to $2m+1$ by Neumaier \cite{Neu81}.

For $W$ as above, we have $a^2/b^2=d/(d+\delta)=r/s$ and $m=(q-1)n$. This immediately implies our claim in the case $s-r \geq 2$ 
(this always happens if $(d,d+\delta)=1$, since $\delta>1$). If
$s=r+1$, we need in addition $r \not\in [1,(\sqrt{2(q-1)n}-1)/2]$ to have again the required bound. \qed

\begin{coro} \label{sph-lin}
In the context of Theorem \ref{thm sc}, if $q$, $n$, $d$, $\delta$, and $k$ are such that
\[ 2(q-1)n+1<q^k, \]
then there exist no linear codes $C \subset E_q^n$ with distances $d$ and $d+\delta$ and dimension at least $k$.
\end{coro}

The tables in Section 5 show that the bound of Theorem \ref{thm sc} is usually better than the best linear programming bound
when $s-r \geq 2$ while it has restricted influence when $s=r+1$.

Corollary \ref{sph-lin} shows, in particular, cases where the linear codes are not the best. For  $q=2$, such cases are
$(n,d,d+\delta)=(11,4,6)$, $(12,4,6)$ and $(15,6,8)$.

\subsection{Some simple cases}

In this section we assume (without loss of generality) that codes under consideration possess the zero word. Then all other words have weights $d$
and $d+\delta$. We start with cases where exactly two distances are impossible, i.e. $A_q(n,\{d,d+\delta\})$ is not well defined. 

\begin{lem} \label{lem-odd-d-even-delta}
There exist no codes $C \subset E_2^n$ with exactly two distances $d$ and $d+\delta$ when:

{\rm (a)} $d$ and $d+\delta$ are both odd; 

{\rm (b)} $d$ is odd, $d+\delta$ is even and $n < (3d-\delta)/2$;

{\rm (c)} $d$ is odd, $d+\delta$ is even and $d<\delta$;

{\rm (d)} $d+\delta=n$ and $n \neq 2d$;

{\rm (e)} $d+\delta=n-1$ and $2d>n+1$.
\end{lem}

\pr Suppose that $|C| \geq 3$ and $x,y \in C \setminus \{\mathbf{0}\}$ are distinct.

(a) Since $d(x,y)=\wt(x)+\wt(y)-2\wt(x*y)$ is even, we obtain a contradiction. 

(b) Using, as in (a) the equality $d(x,y)=\wt(x)+\wt(y)-2\wt(x*y)$ we see that
$\wt(w*y) \in \{(d-\delta)/2,(d+\delta)/2\}$. Therefore
$n \geq d+(d-\max \wt(x*y))=(3d-\delta)/2$, a contradiction. 

(c) If $x$ and $y$ have distinct weights, then $d(x,y)= \wt(x)+\wt(y)-2\wt(x*y)=2d+\delta-2\wt(x*y)$
is odd, thus equal to $d$. Then $2d<d+\delta=2\wt(x*y) \leq 2\min \{\wt(x),\wt(y)\}=2d$, a contradiction. 
If $\wt(x)=\wt(y)=d$, then $d(x,y)$ is even, thus equal to $d+\delta$ and we get a contradiction as above. 

(d) If $x$ and $y$ have distinct weights, then $d(x,y)=n-d \not\in \{d,n\}$ and if $\wt(x)=\wt(y)=d$, then $d(x,y)=n$ is impossible, 
since it leads to $n=2d$. 

(e) Similarly to (d) we see that $d(x,y)=n-d \pm 1 <n-1$ if $x$ and $y$ have distinct weights and $d(x,y)=n-1$ is impossible when 
$\wt(x)=\wt(y)=d$.
\qed

\begin{lem} \label{lem-odd-d=delta}
For $q=2$, if $d$ is odd and $|C|>4$, then
\[ A_2(n,\{d,2d\}) = 1+\left[ \frac{n}{d}\right]. \]
\end{lem}

\pr Let $A_d$ (resp. $A_{2d}$) be the number of the words of weight $d$ (resp. $2d$). Similarly to above we see that if
$\wt(x)=\wt(y)=d$, then $\supp(x) \cap \supp(y)=\phi$. This means that $A_d \leq [n/d]$. Moreover, since $\supp(x) \subset \supp(y)$
for any two words $x$ and $y$ of weights $d$ and $2d$, respectively, it follows that if $A_d \geq 3$, then $A_{2d}=0$, if $A_d=2$, then
$A_{2d}=1$ and $|C|=4$, and if $A_d=1$, then the supports of all words of weight $2d$ contain
the support of the single word of weight $d$ and therefore $A_{2d} \leq [n/d]-1$. In all cases $|C| \leq 1+[n/d]$. It is obvious
from the above how this bound is attained. \qed

\begin{lem} \label{lem-d=1}
We have $A_3(n,\{1,3\}) = 6$ for every $n \geq 4$.
\end{lem}

\pr
Observe that the ternary code $C=\{0000,1000,2110,2120,2201,2202\}$ has distances 1 and 3 and cardinality 6. It can be extended by zero coordinates
to any length $n \geq 4$. Therefore $A_3(n,\{1,3\}) \geq 6$ for $n \geq 4$.

Let $C$ be a maximal $(n,N,\{1,3\})_3$ code of length $n \geq 4$. Without loss of generality we may assume that ${\mathbf 0}$ and $10\ldots 0$ belong to $C$.
Then it is obvious that no more words of weight 1 are possible apart from $20\ldots 0$ in which case $C$ is not maximal.
The words of weight 3 can only have 2 as a first coordinate. We can assume that $x=2110\ldots 0 \in C$. If $y \in C\setminus \{x\}$ has the same support as $x$,
then $d(x,y)=1$ and $y=2210\ldots 0$ without loss of generality and it is easy to see now that no other words with the same support can be added. If  
$z \in C$ has weight 3 and $\supp(x) \neq \supp(z)$, then $d(x,z)=d(y,z)=3$ is only possible and this implies that the first three digits of $z$ are
$202$ and we have only 2 possibilities to complete $z$. This gives $|C| \leq 1+1+2+2=6$ (realized above). If $\supp(x) \neq \supp(y)$ for every $y \in C \setminus \{x\}$ of weight 3, then $d(x,y)=3$ and there are only two possibilities for the first three digits of $y$, $220$ or $202$. Moreover, the third 
nonzero digit of any such $y$ should be at the same position (otherwise distance 2 will appear), which implies that $|C| \leq 5$. 
\qed

The cases covered by Lemmas \ref{lem-odd-d-even-delta}-\ref{lem-d=1} are excluded from the tables below. Other similar
cases can be dealt with as well (for example, one can prove that $A_4(n,\{1,3\})=12$). We formulate as conjectures two observations.

\begin{conj} \label{cc1} (i) $A_q(n,\{2,4\}) = {n \choose 2}+1$ for every $n \geq 6$ and $q \in \{2,3,4\}$;

(ii) $A_2(n,\{2,2+\delta\}) = n$ for every $\delta \geq 3$ and every $n \geq 6$, except for $A_2(2,\{2,n-1\})=n+1$.
\end{conj}

The code consisting of all words of weight 2 and the zero word has distances 2 and 4 and cardinality $ {n \choose 2}+1$. 
This provides the lower bound for (i) which is obviously valid for all $q$.

The lower bound in (ii) is given by the following construction: take all words of length $\delta+3$, weight $\delta+2$ and first 
coordinate 1 (their number is $\delta+2$) and add zero coordinates up to length $n$, then add all words of weight 2 with first
coordinate 1 followed by $\delta+2$ zeros (their number is $n-(\delta+3)$ and finally add the zero word. Note that the code 
consisting of the zero word and all words of weight 2 with nonzero first coordinate gives $n$ words at distance 2 from each 
other and has the same cardinality; if $2+\delta=n-1$, the word of weight $n-1$ with zero first coordinate can be added to get 
 $A_2(2,\{2,n-1\}) \geq n+1$. 

If $\delta$ is odd in (ii), then any two words of weight $\delta+2$ are at even distance, i.e. at distance 2, and therefore share 
common $\delta+1$ nonzero coordinates. The addition of a third word of weight $\delta+2$ gives two possibilities -- a $3 \times (\delta+1)$ or
a $3 \times \delta$ all-ones block. Both determine uniquely the continuation which leads to the above construction in the latter case 
and to a worse cardinality in the former case (note that all words of weight 2, if more than 3, necessarily share the same nonzero
coordinate). The combination of these observations leads to the construction above and proves (ii) for odd $\delta$.

\section{Tables for $A_q(n,\{d,d+\delta\})$}
\label{tables}

\subsection{Randomly generated codes} \label{rand}

We use a simple computer program in Java for random generation of good codes. For fixed length $n$, alphabet size $q$
and distances $d$ and $d+\delta$ the program starts filling into a code $C$ with the zero codeword and the word $(11\ldots100\ldots0)$
of weight $d$. The search space consists of all vectors of weights $d$ and $d+\delta$ of length $n$. Then the program chooses randomly and
adds vectors from the search space until the resulting code is good (i.e., until it has only distances $d$ and
$d+\delta$). The user is able to ask for a verification and finding the distance distribution of all generated codes at any stage of the implementation or later. 

The cardinalities of our best randomly generated codes found along with these of some codes from Lemma 1 are shown as lower bounds in the tables below. 
Some cases of attaining the linear programming bound show that the random approach can be very good; for example, see
the entries $A_2(11,\{2,4\})=56$,  $A_2(18,\{2,4\})=154$,  $A_2(10,\{4,6\})=16$,  $A_2(14,\{4,8\})=64$,
 $A_3(10,\{3,6\})=81$, $A_3(11,\{6,9\})=A_3(12,\{6,9\})=243$, $A_3(12,\{8,12\})=36$, and $A_4(7,\{4,6\})=64$. 

Our Java program and all codes constructed as explained above are available upon request. All best codes will be published on Internet. 

\subsection{Tables}

In this section we present tables with lower and upper bounds for the function $A_q(n,\{d,d+\delta\}$ for
lengths $7 \leq n \leq 20$ for $q=2$, $7 \leq n \leq 14$ for $q=3$ and $7 \leq n \leq 12$ for $q=4$. 
Horizontally we give $d$, vertically $n$. The lower bounds come from computer generated random codes. 
There are three cases of optimal codes from Lemma 1, namely $A_2(8,\{4,8\})=16$, $A_3(9,\{6,9\})=27$, and
$A_4(8,\{6,8\})=32$ (the first two found by our program as well).  In some cases we left an equidistant code since it 
seems quite better than any code with exactly two distances; these codes are marked by $e$. 

The upper bounds are taken from the best of the linear programming bound obtained by the simplex method (unmarked
or marked by $lp$ when seemingly good), the bound \eqref{d2-bound} (marked with $d2$) or its refinements as explained in 
the end of Section 4.2 (marked by $dd$), the best known upper bound on $A_q(n,d)$ [8] (marked with $*$), and the bound 
from Theorem \ref{thm sc} (marked with $sc$).

Key to the tables:

$^{lp}$ -- upper bound by Theorem \ref{thm lp} (general simplex method), excluding cases of Theorem \ref{thm lp1};

$^{*}$ -- upper bound (exact value) from Brouwer's tables \cite{Bro};

$^{sc}$ -- upper bound by Theorem \ref{thm sc} (spherical codes);

$^{d2}$ -- upper bound by Theorem \ref{thm lp1} (special case of Theorem \ref{thm lp});

$^{dd}$ -- contradiction by distance distribution;

$^{e}$ -- this code is equidistant;

$-$ -- $A_q(n,\{d,d+\delta\})$ is not well-defined.

\begin{center}
{\scriptsize
\noindent
\begin{tabular}{|c|c|c|c|c|c|}
\hline
\multicolumn{6}{|c|}{$q=2$, $\delta=2$} \\
\hline
$n|d$ & 2 & 4 & 6 & 8 & 10   \\
\hline
7 &  22-26 & -- &  & & \\ \hline
8 & 29-36 & 10-12$^{d2}$ &  -- & &   \\ \hline
9 & 37-40 & 16$^{d2}$ &  -- & &    \\ \hline
10 & 46-56 &  16$^{d2}$ & 3 &  -- &   \\ \hline
11 & 56$^{lp}$ & 17-23$^{sc}$ & 6-12$^{*,d2}$ & -- &   \\ \hline
12 & 67-77 & 19-25$^{sc}$ & 16$^{d2}$ & -- & --    \\ \hline
13 & 79-87 & 23-40 & 17-19$^{d2}$ & 4* & --  \\ \hline
14 & 92-100 & 27-51 & 17-19$^{d2}$ & $8^{e,*}$ & --   \\ \hline
15 & 106-120 & 32-68 & 18-31$^{sc}$ & 16$^{e,d2}$ & --   \\ \hline
16 & 121-126 & 37-75 & 19-33$^{sc}$ & 17-20$^{d2}$ & --  \\ \hline
17 & 137-154 & 42-91 & 20-35$^{sc}$ & 19-22$^{d2}$ & 6$^{e,lp}$   \\ \hline
18 & 154$^{lp}$ & 46-116 & 20-37$^{sc}$ & 19-22$^{d2}$ & 10$^{e,lp}$  \\ \hline
19 & 172-189 & 52-123 & 21-39$^{sc}$ & 20-35 & 15-20$^{d2}$ \\ \hline
20 & 191-200 & 58-151 & 22-41$^{sc}$ & 20-41$^{sc}$ & 19-24$^{d2}$ \\ \hline
\end{tabular}
}
\end{center}

\begin{center}
{\scriptsize
\noindent
\begin{tabular}{|c|c|c|c|c|c|c|c|c|c|c|}
\hline
\multicolumn{10}{|c|}{$q=2$, $\delta=3$} \\
\hline
$n|d$ & 2 & 4 & 5 & 6 & 7 & 8 & 9 & 10 & 11  \\
\hline
7 & 7-8 & --  &  & & & & & &  \\ \hline
8 & 8-12 & 8$^{e,lp}$ & -- &  & & & & &  \\ \hline
9 & 9-14 & 8$^{e}$-10 & 4$^{lp}$ & -- & & & & &  \\ \hline
10 & 10-18 & 8-16 & 4 & 6$^{e,*}$ & -- & & & &   \\ \hline
11 & 11-19 & 8-16 & 4 & 12$^{e,d2}$ & -- & -- & & &   \\ \hline
12 & 12-24 & 10-21 & 4 & 12$^{e,lp}$ & 3 & -- & -- & &   \\ \hline
13 & 13-24 & 12-27$^{sc}$ & 4 & 14$^{lp}$ & 4 & -- & -- & -- &   \\ \hline
14 & 14-28 & 14-29$^{sc}$ & 5-8 & 14-27 & 4 & -- & -- & -- & --  \\ \hline
15 & 15-28 & 14-31$^{sc}$ & 7-16 & 14-27 & 6$^{lp}$ & 16$^{e,d2}$ & -- & -- & --  \\ \hline
16 & 16-32 & 14-33$^{sc}$ & 7-16 & 15-34 & 6$^{lp}$ & 16$^{e,lp}$ & 4$^{lp}$ & -- & --   \\ \hline
17 & 17-33 & 14-35$^{sc}$ & 8-18 & 15-50 & 6$^{lp}$ & 17-21 & 4-6 & -- & --   \\ \hline
18 & 18-36 & 14-37$^{sc}$ & 10-22 & 16-65 & 7-10$^{lp}$ & 17-29 & 8$^{lp}$ & 4 & 4$^{*,lp}$   \\ \hline
19 & 19-37 & 14-39$^{sc}$ & 13-35 & 16-70 & 9-20$^{lp}$ & 17-29 & 8$^{lp}$ & 20$^{e,lp}$ & 4$^{lp}$   \\ \hline
20 & 20-40 & 14-41$^{sc}$ & 17-41 & 16-89 & 11-20$^{lp}$ & 20-41$^{sc}$ & 8$^{lp}$ & 20$^{e,lp}$ & 4   \\ \hline
\end{tabular}
}
\end{center}

\begin{center}
{\scriptsize
\noindent
\begin{tabular}{|c|c|c|c|c|c|}
\hline
\multicolumn{6}{|c|}{$q=2$, $\delta=4$} \\
\hline
$n|d$ & 2 & 4 & 6 & 8 & 10    \\
\hline
7 & 8$^{lp}$ &  & & &    \\ \hline
8 & 8$^{lp}$ &16$^{*,d2}$ & & &     \\ \hline
9 & 9-16 & 16$^{lp}$ & & &      \\ \hline
10 & 10-18 & 16$^{lp}$ & -- & &   \\ \hline
11 & 12-23$^{sc}$ & 16-30 & 12$^{e,d2}$ & &    \\ \hline
12 & 12-25$^{sc}$ & 16-30 & 12-19$^{dd}$ & -- &    \\ \hline
13 & 13-27$^{sc}$ & 32-54 & 13-27$^{sc}$ & -- &    \\ \hline
14 & 14-29$^{sc}$ & 64$^{lp}$ & 14-29$^{sc}$ & 8$^{e,*}$ & --   \\ \hline
15 & 15-31$^{sc}$ & 64-88 & 16-31$^{sc}$ & 16$^{e,d2}$ & --    \\ \hline
16 & 16-33$^{sc}$ & 64-128 & 16-33$^{sc}$ & 21-24$^{d2}$ & --   \\ \hline
17 & 17-35$^{sc}$ & 64-150 & 17-35$^{sc}$ & 32-36 & 6$^{e,*}$   \\ \hline
18 & 18-37$^{sc}$ & 64-256 & 18-37$^{sc}$ & 64$^{d2}$ & 10$^{e,lp}$   \\ \hline
19 & 19-39$^{sc}$ & 64-256 & 20-39$^{sc}$ & 80-96$^{d2}$ & 20$^{e,lp}$   \\ \hline
20 & 20-41$^{sc}$ & 64-332 & 20-41$^{sc}$ & 80-96$^{d2}$ & 20-27$^{dd}$  \\ \hline
\end{tabular}
}
 \ \ \ 
{\scriptsize
\noindent
\begin{tabular}{|c|c|c|c|c|c|c|}
\hline
\multicolumn{6}{|c|}{$q=2$, $\delta=6$} \\
\hline
$n|d$ & 2 & 4 & 6 & 8 & 10    \\
\hline
9 & 10-16 &  & & &  \\ \hline
10 & 10-16 & --  &  & &  \\ \hline
11 & 11-23$^{sc}$ & 8$^{e}$-16  &  & &  \\ \hline
12 & 12-25$^{sc}$ & 8$^{e}$-19 &  -- & & \\ \hline
13 & 13-27$^{sc}$ & 8$^{e}$-26 &  24$^{lp}$ & &    \\ \hline
14 & 14-29$^{sc}$ & 8-26 &  24-26 & -- &  \\ \hline
15 & 15-31$^{sc}$ & 11-31$^{sc}$ & 24-27 & 16$^{e,d2}$ &  \\ \hline
16 & 16-33$^{sc}$ & 14-33$^{sc}$ & 24-29 & 16$^{e}$-27$^{dd}$ & --   \\ \hline
17 & 17-35$^{sc}$ & 14-35$^{sc}$ & 24-52 & 17-35$^{sc}$ & --  \\ \hline
18 & 18-37$^{sc}$ & 16-37$^{sc}$ & 24-52 & 17-37$^{sc}$ & 10$^{e,lp}$   \\ \hline
19 & 19-39$^{sc}$ & 18-39$^{sc}$ & 28-68 & 17-39$^{sc}$ & 20$^{e,lp}$   \\ \hline
20 & 20-41$^{sc}$ & 20-41$^{sc}$ & 48-123 & 17-41$^{sc}$ & 20-32   \\ \hline
\end{tabular}
}
\end{center}

\begin{center}
{\scriptsize
\noindent
\begin{tabular}{|c|c|c|c|c|c|c|c|c|c|}
\hline
\multicolumn{9}{|c|}{$q=2$, $\delta=5$} \\
\hline
$n|d$ & 2 & 4 & 6 & 7 & 8 & 9 & 10 & 11   \\
\hline
8 & 9-10   &  &  & & & & &    \\ \hline
9 & 9-10 &  --   &  & & & & &    \\ \hline
10 & 10-16 & 8$^{e}$-16 &  &  & & & &   \\ \hline
11 & 11-18 & 8$^{e}$-16 & -- &  &  & & &   \\ \hline
12 & 12-24 & 8-25 & 12$^{e,lp}$ & -- &  &  & &    \\ \hline
13 & 13-24 & 8-25 & 13$^{e}$-14 & 4$^{lp}$ & -- & &  &   \\ \hline
14 & 14-28 & 10-29$^{sc}$ &  13$^{e}$-19 & 4$^{lp}$ & -- & -- & &    \\ \hline
15 & 15-29 & 14-31$^{sc}$ &  14$^{e}$-28 & 4$^{lp}$ & 16$^{e,d2}$ & -- & -- &   \\ \hline
16 & 16-32 & 14-33$^{sc}$ & 14$^{e}$-33$^{sc}$ & 4$^{lp}$ & 16$^{e,lp}$ & 4$^{lp}$ & -- & --   \\ \hline
17 & 17-34 & 16-35$^{sc}$ &  14-35$^{sc}$ & 4$^{lp}$ & 16$^{e}$-18$^{lp}$ & 4$^{lp}$ & 6$^{e,*}$ & --  \\ \hline
18 & 18-36 & 18-37$^{sc}$ &  15-37$^{sc}$ & 4$^{lp}$ & 17$^{e}$-22$^{lp}$ & 4 & 10$^{e,*}$ & --   \\ \hline
19 & 19-38 & 18-39$^{sc}$ &  15-39$^{sc}$ & 4$^{lp}$ & 17$^{e}$-35$^{lp}$ & 4 & 20$^{e,*}$ & 4$^{lp}$   \\ \hline
20 & 20-40 & 18-41$^{sc}$ &  16-41$^{sc}$ & 4-6$^{lp}$ & 21-41$^{sc}$ & 4 & 20$^{e,lp}$ & 4$^{lp}$   \\ \hline
\end{tabular}
}
\end{center}

\begin{center}
{\scriptsize
\noindent
\begin{tabular}{|c|c|c|c|c|c|c|c|c|c|c|c|}
\hline
\multicolumn{12}{|c|}{$q=3$, $\delta=2$} \\
\hline
$n|d$ & 2 & 3 & 4 & 5 & 6 & 7 & 8 & 9 & 10 & 11 & 12  \\
\hline
7 & 22-57 & 13-21 & 19-28 & 8-15$^{d2}$ & & & & & & &   \\ \hline
8 & 29-81 & 13-23 & 19-37 & 17-28 & 9$^{e,*}$ & & & & & &   \\ \hline
9 & 37-86 & 13-30 & 19-57 & 17-35$^{d2}$ & 16-24$^{d2}$  & 6$^{e,*}$ & & & & &  \\ \hline
10 & 46-111 & 13-33 & 20-81 & 17-41$^{sc}$ & 28-36$^{d2}$ & 13-21$^{d2}$ & 6$^{e,*}$ & & & &  \\ \hline
11 & 56-158 & 13-33 & 20-125 & 17-45$^{sc}$ & 28-45 & 16-33$^{d2}$ & 12$^{e,*}$ & 4$^{e,*}$ & & &   \\ \hline
12 & 67-197 & 13-33 & 20-162 & 17-49$^{sc}$ & 28-49$^{sc}$ & 18-37$^{d2}$ & 18-30$^{d2}$ & 9$^{e,*}$ & 4$^{e,*}$ & &   \\ \hline
13 & 79-204 & 13-33 & 21-259 & 17-53$^{sc}$ & 28-53$^{sc}$ & 18-53$^{sc}$ & 19-40 & 18-27$^{d2}$ & 6$^{e,*}$ & 3* &   \\ \hline
14 & 92-249 & 13-33 & 21-275 & 17-57$^{sc}$ & 28-57$^{sc}$ & 18-57$^{sc}$ & 20-46 & 18-38$^{d2}$ & 11-15 & 6$^{e,*}$ & 3*   \\ \hline
\end{tabular}
}
\end{center}

\begin{center}
{\scriptsize
\noindent
\begin{tabular}{|c|c|c|c|c|c|c|c|c|c|c|c|}
\hline
\multicolumn{11}{|c|}{$q=3$, $\delta=3$} \\
\hline
$n|d$ & 1 & 2 & 3 & 4 & 5 & 6 & 7 & 8 & 9 & 10   \\
\hline
7 & 9-20 & 7-29 & 27$^{lp}$ & 9-27 & & & & & &     \\ \hline
8 & 9-28 & 8-33$^{sc}$ & 81$^{lp}$ & 9-33$^{sc}$ & 9-33$^{sc}$ & & & & &    \\ \hline
9 & 9-37$^{sc}$ & 9-37$^{sc}$ & 81$^{lp}$ & 9-37$^{sc}$ & 10-37$^{sc}$ & 27$^{*,d2}$  & & & &   \\ \hline
10 & 10-41$^{sc}$ & 10-41$^{sc}$ & 81$^{lp}$ & 10-41$^{sc}$ & 10-41$^{sc}$ & 81$^{*,d2}$ & 10-21 & & &    \\ \hline
11 & 10-45$^{sc}$ & 11-45$^{sc}$ & 81-91 & 12-45$^{sc}$ & 12-45$^{sc}$ & 243$^{*,d2}$ & 12-45$^{d2}$ & 12$^{e,*}$ & &    \\ \hline
12 & 12-49$^{sc}$ & 12-49$^{sc}$ & 81-106 & 12-49$^{sc}$ & 13-49$^{sc}$ & 243$^{lp}$ & 13-49$^{sc}$ & 13-33$^{d2}$ & 9$^{e,*}$ &   \\ \hline
13 & 12-53$^{sc}$ & 13-53$^{sc}$ & 81-139 & 13-53$^{sc}$ & 13-53$^{sc}$ & 243-448 & 15-53$^{sc}$ & 13-53$^{sc}$ & 24-27$^{d2}$ & 6$^{e,*}$   \\ \hline
14 & 14-57$^{sc}$ & 14-57$^{sc}$ & 81-162 & 14-57$^{sc}$ & 13-57$^{sc}$ & 243-729 & 15-57$^{sc}$ & 15-57$^{sc}$ & 31-57$^{sc}$ & 11-48$^{d2}$ \\ \hline
\end{tabular}}
\end{center}

\begin{center}
{\scriptsize
\noindent
\begin{tabular}{|c|c|c|c|c|c|c|c|c|c|c|c|}
\hline
\multicolumn{11}{|c|}{$q=3$, $\delta=4$} \\
\hline
$n|d$ & 1 & 2 & 3 & 4 & 5 & 6 & 7 & 8 & 9 & 10  \\
\hline
7 & 6-21 & 8-20 & 9$^{e}$-13 & & & & & & & \\ \hline
8 & 9-22 & 10-33$^{sc}$ & 10-21 & 17-23 & & & & & & \\ \hline
9 & 10-29 & 12-37$^{sc}$ & 15-32 & 18-51 & 9-33 &  & & & &  \\ \hline
10 & 12-33 & 12-41$^{sc}$ & 16-41$^{sc}$ & 36-61 & 10-41$^{sc}$ & 15$^{e}$-41$^{sc}$ & & & & \\ \hline
11 & 14-33 & 12-45$^{sc}$ & 16-45$^{sc}$ & 42-144 & 12-45$^{sc}$ & 15$^{e}$-45$^{sc}$ & 12-45$^{sc}$ & & &  \\ \hline
12 & 16-33 & 17-49$^{sc}$ & 16-49$^{sc}$ & 49-195 & 13-49$^{sc}$ & 22-49$^{sc}$ & 25-49$^{sc}$ & 36$^{d2}$ & &  \\ \hline
13 & 18-33 & 18-53$^{sc}$ & 16-53$^{sc}$ & 56-317 & 19-53$^{sc}$ & 22-53$^{sc}$ & 25-53$^{sc}$ & 36-85 & 27$^{e,d2}$ & \\ \hline
14 & 18-33 & 18-57$^{sc}$ & 16-57$^{sc}$ & 56-557 & 19-57$^{sc}$ & 22-57$^{sc}$ & 25-57$^{sc}$ & 36-108 & 27$^{e}$-57$^{sc}$ & 12-15 \\ \hline
\end{tabular}
}
\end{center}

\begin{center}
{\scriptsize
\noindent
\begin{tabular}{|c|c|c|c|c|c|c|c|c|c|c|}
\hline
\multicolumn{10}{|c|}{$q=3$, $\delta=5$} \\
\hline
$n|d$ & 1 & 2 & 3 & 4 & 5 & 6 & 7 & 8 & 9  \\
\hline
7 & 6-13 & 8-21 &  & &  & & & &  \\ \hline
8 & 6-23 & 9-33$^{sc}$ & 9$^{e}$-15 & &  & & & &  \\ \hline
9 & 10-34 & 9-37$^{sc}$ & 10-25 & 9-29 &  & & & &   \\ \hline
10 & 12-41$^{sc}$ & 12-41$^{sc}$ & 15-39 & 10-41$^{sc}$ & 9-33 & & & &   \\ \hline
11 & 14-45$^{sc}$ & 12-45$^{sc}$ & 17-45$^{sc}$ & 12-45$^{sc}$ & 12-45 & 15$^{e}$-45 & & &  \\ \hline
12 & 17-49$^{sc}$ & 12-49$^{sc}$ & 17-49$^{sc}$ & 13-49$^{sc}$ & 21-75 & 18$^{e}$-45 & 13-49$^{sc}$ & &   \\ \hline
13 & 17-53$^{sc}$ & 13-53$^{sc}$ & 17-53$^{sc}$ & 19-53$^{sc}$ & 31-140 & 22-53$^{sc}$ & 13-53$^{sc}$ & 13$^{e}$-53$^{sc}$ &   \\ \hline
14 & 17-57$^{sc}$ & 14-57$^{sc}$ & 17-57$^{sc}$ & 19-57$^{sc}$ & 50-271 & 22-57$^{sc}$ & 15-57$^{sc}$ & 17-57$^{sc}$ & 17-57$^{sc}$   \\ \hline
\end{tabular}
}
\end{center}

\begin{center}
{\scriptsize
\noindent
\begin{tabular}{|c|c|c|c|c|c|c|c|c|c|}
\hline
\multicolumn{9}{|c|}{$q=3$, $\delta=6$} \\
\hline
$n|d$ & 1 & 2 & 3 & 4 & 5 & 6 & 7 & 8  \\
\hline
7 & 4-7 & &  & &  & & &   \\ \hline
8 & 6-10 & 9-22 & & &  & & &   \\ \hline
9 & 6-17 & 10-34 & 9$^{e}$-18 & &  & & &   \\ \hline
10 & 8-41 & 10-41$^{sc}$ & 10-23 & 9-41$^{sc}$ & & & &   \\ \hline
11 & 10-42 & 12-45$^{sc}$ & 16-32 & 10-45$^{sc}$ & 8$^{e}$-35 & & &   \\ \hline
12 & 13-49$^{sc}$ & 13-49$^{sc}$ & 19-49$^{sc}$ & 13-49$^{sc}$ & 9-49$^{sc}$ & 25-45 & &   \\ \hline
13 & 13-53$^{sc}$ & 13-53$^{sc}$ & 24-53$^{sc}$ & 13-53$^{sc}$ & 10-53$^{sc}$ & 25-68 & 13-53$^{sc}$ &   \\ \hline
14 & 14-57$^{sc}$ & 14-57$^{sc}$ & 24-57$^{sc}$ & 14-57$^{sc}$ & 13-57$^{sc}$ & 28-106 & 14-57$^{sc}$ & 15-57$^{sc}$  \\ \hline
\end{tabular}
}
\end{center}

\begin{center}
{\scriptsize
\noindent
\begin{tabular}{|c|c|c|c|c|c|c|c|c|c|c|c|c|c|c|c|}
\hline
\multicolumn{11}{|c|}{$q=4$, $\delta=2$} \\
\hline
$n|d$ & 1 & 2 & 3 & 4 & 5 & 6 & 7 & 8 & 9 & 10 \\
\hline
7 & 12-28 & 22-64 & 13-43$^{sc}$ & 64$^{d2}$ & 14-32 & & & & &   \\ \hline
8 & 12-28 & 29-112 & 13-49$^{sc}$ & 64-146 & 17-49$^{sc}$ & 32$^{*,d2}$ & & & &   \\ \hline
9 & 12-28 & 37-179 & 14-55$^{sc}$ & 64-179 & 17-55$^{sc}$ & 59-64$^{d2}$  & 14-20* & & & \\ \hline
10 & 12-28 & 46-256 & 16-61$^{sc}$ & 64-290 & 17-61$^{sc}$ & 59-89 & 19-56$^{d2}$ & 16$^{e,*}$ & &  \\ \hline
11 & 12-28 & 56-320 & 16-67$^{sc}$ & 64-358 & 17-67$^{sc}$ & 59-179 & 20-56$^{d2}$ & 28-49$^{d2}$ & 12$^{e,*}$ &   \\ \hline
12 & 12-28 & 67-320 & 16-73$^{sc}$ &  64-526 & 17-73$^{sc}$ & 59-213 & 20-73$^{sc}$ & 37-64$^{d2}$ & 18-44$^{d2}$ & 8-9*   \\ \hline
\end{tabular}
}
\end{center}

\begin{center}
{\scriptsize
\noindent
\begin{tabular}{|c|c|c|c|c|c|c|c|c|c|c|}
\hline
\multicolumn{10}{|c|}{$q=4$, $\delta=3$} \\
\hline
$n|d$ & 1 & 2 & 3 & 4 & 5 & 6 & 7 & 8 & 9  \\
\hline
7 & 14-43$^{sc}$ & 12-43$^{sc}$ & 36-52 & 16-31 & & & & &   \\ \hline
8 & 14-49$^{sc}$ & 12-49$^{sc}$ & 81-113 & 18-49$^{sc}$ & 16-49$^{sc}$ & & & &   \\ \hline
9 & 14-55$^{sc}$ & 12-55$^{sc}$ & 81-270 & 19-55$^{sc}$ & 16-55$^{sc}$ & 28-76  & & &  \\ \hline
10 & 15-61$^{sc}$ & 13-61$^{sc}$ & 81-352 & 19-61$^{sc}$ & 16-61$^{sc}$ & 39-216 & 16-61$^{sc}$ & &  \\ \hline
11 & 15-67$^{sc}$ & 14-67$^{sc}$ & 81-511 & 19-67$^{sc}$ & 16-67$^{sc}$ & 47-320 & 17-67$^{sc}$ & 16-60$^{*}$ &  \\ \hline
12 & 18-73$^{sc}$ & 14-73$^{sc}$ & 81-738 & 19-73$^{sc}$ & 16-73$^{sc}$ & 47-779 & 18-73$^{sc}$ & 17-73$^{sc}$ & 15-48$^{d2}$  \\ \hline
\end{tabular}}
\end{center}

\begin{center}
{\scriptsize
\noindent
\begin{tabular}{|c|c|c|c|c|c|c|c|c|c|}
\hline
\multicolumn{9}{|c|}{$q=4$, $\delta=4$} \\
\hline
$n|d$ & 1 & 2 & 3 & 4 & 5 & 6 & 7 & 8  \\
\hline
7 & 14-43$^{sc}$ & 12-43$^{sc}$ & 12-40 & & & & &  \\ \hline
8 & 17-49$^{sc}$ & 14-49$^{sc}$ & 17-49$^{sc}$ & 32-38 & & & & \\ \hline
9 & 17-55$^{sc}$ & 14-55$^{sc}$ & 17-55$^{sc}$ & 64-82 & 12-44 &  & &  \\ \hline
10 & 20-61$^{sc}$ & 16-61$^{sc}$ & 17-61$^{sc}$ & 256-298 & 19-61$^{sc}$ & 20-58 & &  \\ \hline
11 & 22-67$^{sc}$ & 20-67$^{sc}$ & 17-67$^{sc}$ & 256-353 & 19-67$^{sc}$ & 20-67$^{sc}$ & 13-67$^{sc}$ &  \\ \hline
12 & 25-73$^{sc}$ & 25-73$^{sc}$ & 17-73$^{sc}$ & 256-656 & 19-73$^{sc}$ & 20-73$^{sc}$ & 14-73$^{sc}$ & 28-73$^{sc}$  \\ \hline
\end{tabular}
}
\end{center}

\begin{center}
{\scriptsize
\noindent
\begin{tabular}{|c|c|c|c|c|c|c|c|c|}
\hline
\multicolumn{8}{|c|}{$q=4$, $\delta=5$} \\
\hline
$n|d$ & 1 & 2 & 3 & 4 & 5 & 6 & 7  \\
\hline
7 & 9-28 & 14-43$^{sc}$ &  & &  & &   \\ \hline
8 & 13-49$^{sc}$ & 14-49$^{sc}$ & 18-49$^{sc}$ & &  & &   \\ \hline
9 & 17-55$^{sc}$ & 16-55$^{sc}$ & 18-55$^{sc}$ & 16-55$^{sc}$ &  & &   \\ \hline
10 & 21-61$^{sc}$ & 18-61$^{sc}$ & 18-61$^{sc}$ & 17-61$^{sc}$ & 16-82 & &   \\ \hline
11 & 21-67$^{sc}$ & 18-67$^{sc}$ & 18-67$^{sc}$ & 20-67$^{sc}$ & 32-132 & 24-29 &   \\ \hline
12 & 22-73$^{sc}$ & 18-73$^{sc}$ & 18-73$^{sc}$ & 20-73$^{sc}$ & 36-323 & 24-73$^{sc}$ & 14-73$^{sc}$   \\ \hline
\end{tabular}
}
\end{center}

\begin{center}
{\scriptsize
\noindent
\begin{tabular}{|c|c|c|c|c|c|c|c|}
\hline
\multicolumn{7}{|c|}{$q=4$, $\delta=6$} \\
\hline
$n|d$ & 1 & 2 & 3 & 4 & 5 & 6 \\
\hline
7 & 6-14 & &  & & &  \\ \hline
8 & 9-24 & 16-49$^{sc}$ & & & &   \\ \hline
9 & 12-55$^{sc}$ & 16-55$^{sc}$ & 18-55$^{sc}$ & & &   \\ \hline
10 & 15-61$^{sc}$ & 18-61$^{sc}$ & 18-61$^{sc}$ & 16-61$^{sc}$ & &  \\ \hline
11 & 17-67$^{sc}$ & 18-67$^{sc}$ & 18-67$^{sc}$ & 19-67$^{sc}$ & 12$^{e}$-67$^{sc}$ &  \\ \hline
12 & 24-73$^{sc}$ & 18-73$^{sc}$ & 18-73$^{sc}$ & 19-73$^{sc}$ & 14-73$^{sc}$ & 48-152  \\ \hline
\end{tabular}
}
\end{center}

{\bf Acknowledgements.} The first author was partially supported
by the National Scientific Program "Information and Communication
Technologies for a Single Digital Market in Science, Education and
Security (ICTinSES)" of the Bulgarian Ministry of
Education and Science. The second author was supported by the National Programme "Young Scientists and PostDocs"  of the Bulgarian Ministry of
Education and Science. The research of the third and forth authors was carried out at the IITP RAS at the expense
of the Russian Fundamental Research Foundation (project No.
19-01-00364). The authors thank the anonymous reviewers for their careful reading and useful remarks.

\end{document}